\documentclass{l4dc2025}

\usepackage{mathtools}
\usepackage{xcolor}

\usepackage{comment}

\title{Predictive Monitoring of Black-Box Dynamical Systems}
\usepackage{times}
\usepackage{booktabs}
\usepackage{multirow}
\usepackage{adjustbox}
\usepackage{siunitx}
\usepackage{enumerate,paralist}
\usepackage{comment}
\usepackage[utf8]{inputenc}
\usepackage[english]{babel}
\usepackage{physunits}
\usepackage{wrapfig}

\author{%
	\Name{Thomas A. Henzinger}\textsuperscript{1} \Email{tah@ist.ac.at} \\
	\Name{Fabian Kresse}\textsuperscript{1} \Email{fabian.kresse@ist.ac.at} \\
		\Name{Kaushik Mallik}\textsuperscript{2}  \Email{kmallik314@gmail.com} \\
	\Name{Emily Yu}\textsuperscript{1} \Email{emily.yu@ist.ac.at} \\
	\Name{\DJ or\dj e \v{Z}ikeli\'c}\textsuperscript{3} \Email{dzikelic@smu.edu.sg}\\
	\addr \textsuperscript{1}Institute of Science and Technology Austria, Klosterneuburg, Austria\\
		\addr \textsuperscript{2} IMDEA Software Institute, Madrid, Spain \\
			\addr \textsuperscript{3} Singapore Management University, Singapore, Singapore \\
}

\newtheorem{assumption}{Assumption}

\newcommand{\tup}[1]{(#1)}

\newcommand{\Sys}{\xi}
\newcommand{\est}[2]{\overline{#1_{#2}}}
\newcommand{\grad}[3]{\nabla^{#1}_{#2}\, #3}

\newcommand{\spec}{\varphi}
\newcommand{\val}{\spec}

\begin{document}

\maketitle

\begin{abstract}
We study the problem of predictive runtime monitoring of black-box dynamical systems with quantitative safety properties. The black-box setting stipulates that the exact semantics of the dynamical system and the controller are unknown, and that we are only able to observe the state of the controlled (aka, closed-loop) system at finitely many time points. We present a novel framework for predicting future states of the system based on the states observed in the past. The numbers of past states and of predicted future states are parameters provided by the user. Our method is based on a combination of Taylor's expansion and the backward difference operator for numerical differentiation. We also derive an upper bound on the prediction error under the assumption that the system dynamics and the controller are smooth. The predicted states are then used to predict safety violations ahead in time. Our experiments demonstrate practical applicability of our method for complex black-box systems, showing that it is computationally lightweight and yet significantly more accurate than the state-of-the-art predictive safety monitoring techniques. 
\end{abstract}


%!TEX root=main.tex

\section{Introduction}\label{sec:intro}

A majority of autonomous systems nowadays depend on advanced artificial intelligence (AI) technologies.
For instance, in self-driving cars, perception modules are almost synonymous to machine-learned computer vision software~\citep{JanaiGBG20} and controllers are routinely designed using deep reinforcement learning algorithms~\citep{LillicrapHPHETS15}.
Although revolutionary, these AI technologies are hard to analyze and pose serious risk with respect to the safe and correct behavior of the underlying systems~\citep{AmodeiOSCSM16}.

% ------- Djordje's opening ------------
%Many autonomous systems, such as autonomous vehicles, are powered by advanced artificial intelligence (AI) technologies. For instance, perception modules in self-driving cars regularly utilize computer vision technologies~\citep{JanaiGBG20}. On the other hand, learning-based methods such as deep reinforcement learning and the use of neural networks as controllers have shown great promise in solving hard control tasks in autonomous systems~\citep{LillicrapHPHETS15}. However, while bringing significant value, these AI-empowered technologies also introduce new challenges in the context of ensuring their correct and safe behaviour~\citep{AmodeiOSCSM16}. %The safety-critical nature of many autonomous systems, e.g.~self-driving cars or healthcare devices, requires a high level of correctness assurances in order to allow for their safe and trustworthy deployment.
% ---------------------------

Towards the safe and trustworthy deployment of AI-powered systems, we study the problem of {\em predictive runtime monitoring} of continuous-time dynamical systems possibly operated by a learned controller. We consider the {\em black-box setting}, in which the exact semantics of the system dynamics and the controller are unknown. Rather, one is only able to observe the state of the system at finitely many sampling instances. Our goal is to design a runtime monitoring algorithm which, at each sampling point, takes the past states into account, and \emph{predicts} the future states and resulting future safety status of the system within a given time horizon. Unlike traditional runtime monitoring concerning fulfillment of safety only in the \emph{past}~\citep{bartocci2018lectures}, our predictive monitors raise safety warnings \emph{before} they actually take place, so that the system can be intervened and steered out of danger in time, e.g., by using a fail-safe backup controller as in shielding~\citep{AlshiekhBEKNT18}.

% ------------------------------------------
%Runtime monitoring is a lightweight technique for ensuring system correctness by monitoring the system at runtime and reasoning about observed system states ~\citep{bartocci2018lectures}. It has established itself as a valuable complement to the quality assurance methods used at system design time. However, most traditional runtime monitoring methods use observed system states to reason about {\em past violations} of the safety specification of interest. In contrast, our problem is concerned with predicting {\em future states and violations} by reasoning about past observed states. 
%------------------------------------------

\begin{wrapfigure}{r}{0.4\linewidth}
	\centering
	\includegraphics[trim={1cm 1cm 1cm 2.5cm},clip,width=\linewidth]{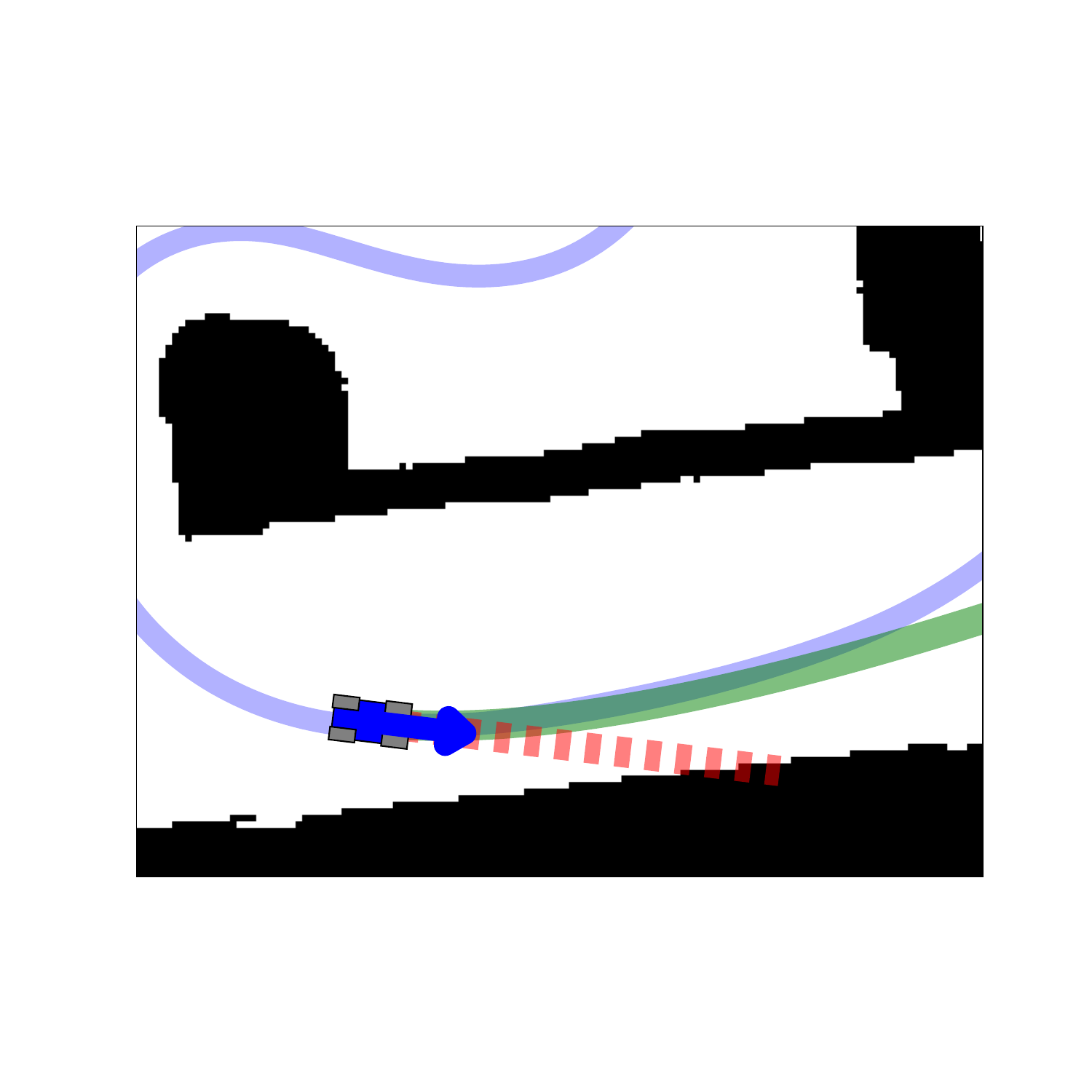}\\
	\includegraphics[trim={1cm 1cm 0.5cm 2.5cm},clip,width=0.85\linewidth]{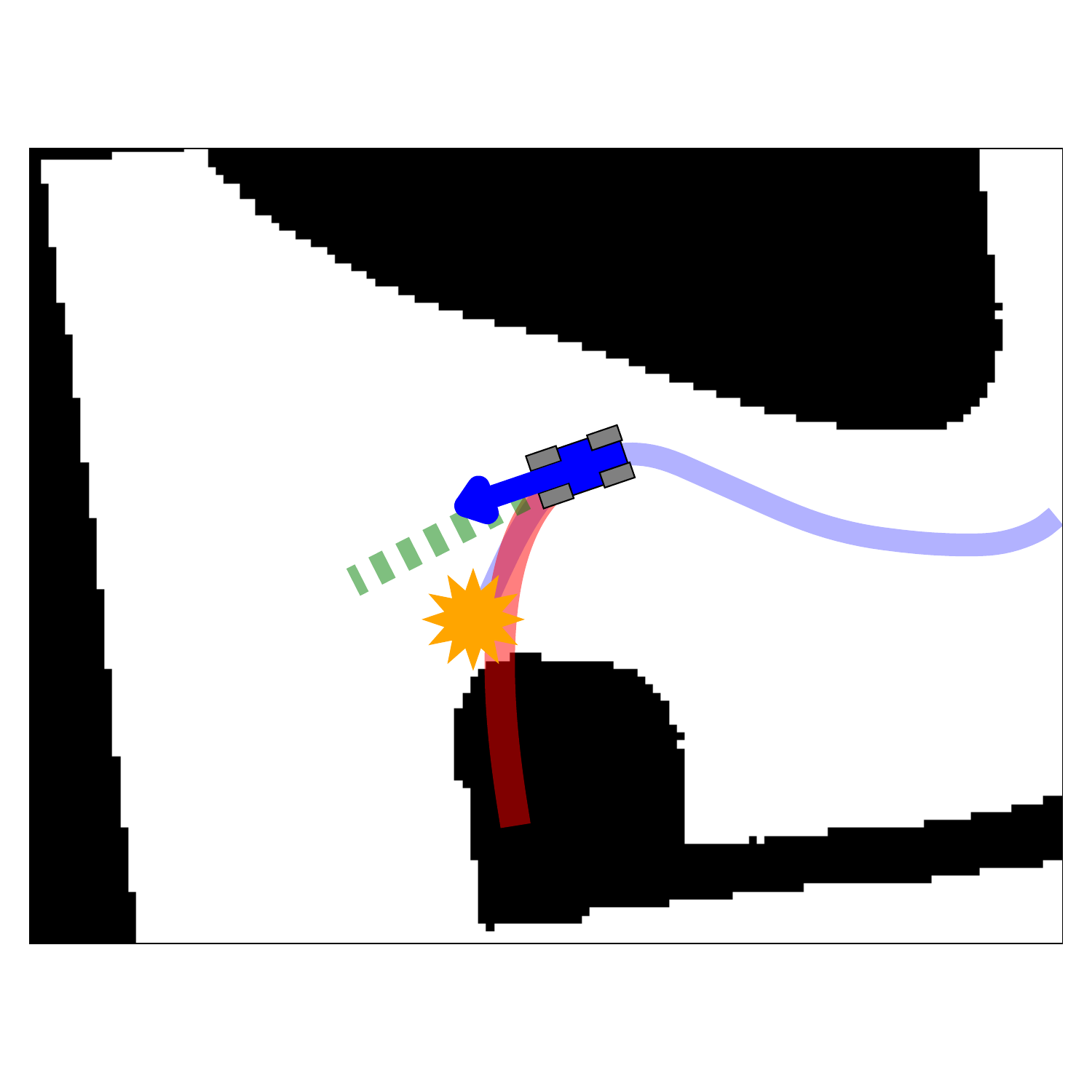}
	\caption{Predictive runtime monitoring using our TPM (solid) and the baseline TTC (dashed). The solid blue line is the ground truth trajectory. Green and red represent predictions of, respectively, safe and unsafe behaviors within the horizon. As can be seen, TMP is more accurate in predicting smooth turns.}
	\label{fig:motivating_example_race_car}
\end{wrapfigure}

\paragraph{Our contributions.} We present {\em Taylor-based Predictive Monitoring (TPM)}, a new framework for predictive runtime monitoring of black-box controlled dynamical systems. For a given time $t$ and past observed states $x_{t-k\tau},\dots,x_{t-\tau},x_t$ evenly sampled with a given time interval $\tau$, our goal is to predict the next $h$ states $x_{t+\tau},\dots, x_{t+h\tau}$. Here, the numbers $k+1$ of past observed states and $h$ of future states to be predicted are parameters that can be chosen by the user. %Our method makes mild assumptions about the system dynamics and the controller, assuming only that they are smooth and locally Lipschitz continuous functions. 

TPM consists of a {\em learning} phase followed by a {\em prediction} phase. In the learning phase, TPM first uses Taylor's polynomials in order to approximate the true system dynamics and the controller via a polynomial function. Taylor's polynomial expansion is a classical result in mathematical analysis that allows us to represent an arbitrary nonlinear function around a given point using polynomials with arbitrary accuracy~\citep{rudin1964principles}. The challenge is that the coefficients of Taylor's polynomials use derivatives of the function, which in our setting are unknown owing to the black-box nature of the system. To address this challenge, we then use the backward difference method~\citep{gear1967numerical} to numerically approximate the derivatives from the past state observations. By combining these two ingredients, we obtain an approximate polynomial model of the system in the vicinity of the current time. In the prediction phase, this polynomial model is used to compute the predicted future states $x_{t+\tau},\dots, x_{t+h\tau}$. Finally, the predictions can be used to reason about possible future violations of a safety specification of interest. TPM can reason about both {\em qualitative} and more general {\em quantitative} safety specifications, see Sec.~\ref{prelims} for formal definition.

Taylor's polynomials and backward difference method are both standard in numerical analysis. However, to the best of our knowledge, their combined application to monitoring black-box controlled dynamical systems is novel. We also derive a {\em formal upper bound on the approximation error} induced by TPM. Our formal analysis assumes that the system dynamics and the controller are $(l+1)$-times continuously differentiable functions, where $l$ is the degree of Taylor's polynomial and equals to $k$ in our case. %The error bound then allows choosing a suitable degree $l$ of the Taylor's expansion in order to control the precision of the approximation.
%However, this assumption is only necessary for the formal error analysis, and our approximation procedure remains well defined even in the case of non-differentiable system dynamics or controllers.

We implemented and experimentally evaluated TPM on two complex controlled dynamical systems. 
As a baseline, we compared our method to the time-to-collision (TTC) metric~\citep{vogel2003comparison}, a technique for providing runtime safety assurances of autonomous systems. TTC is routinely deployed in autonomous driving to predict on-road safety violation~\citep{wang2021review}, and is defined as the time after which the vehicle will violate safety (i.e.~cause ``collision'') if it continues with its current velocity. One can view TTC as a special case of our method where only degree~$1$ Taylor's polynomial is used. As shown in Fig.~\ref{fig:motivating_example_race_car}, reasoning about the first-order derivative only can be too conservative and may lead to failures in correctly predicting future safety violations. This is seen in our experiments as well, with TPM showing superior predictive power compared to TTC.

%\begin{figure}[t]
%	\includegraphics[scale=0.3]{figures/intro.pdf}
%	\caption{Predictive monitoring using the Taylor's polynomial approach and the state-of-the-art time to collision (TTC) approach: the former is more accurate in predicting smooth turns.}
%	\label{fig:motivating example:race car}
%\end{figure}

Our contributions can be summarized as follows:
\begin{compactenum}
	\item {\em Predictive runtime monitoring.} We present Taylor-based Predictive Monitoring (TPM), a framework for predictive runtime monitoring of black-box controlled dynamical systems with a given quantitative safety specification. TPM is based on a combination of Taylor's polynomials and the backward difference method of numerical differentiation.
	\item {\em Formal error bound analysis.} We provide a formal error bound analysis for our predictions.
	\item {\em Experiments.} Results demonstrate practical applicability of our framework to complex dynamical systems. TPM shows superior predictive power compared to the baseline TTC, a classical approach for providing runtime assurances about autonomous systems correctness.
\end{compactenum}

\medskip
\noindent\textbf{Related Work.}
In the formal methods literature, traditional runtime verification approaches treat the monitored system as a black-box, and output at each time point whether a given specification has been violated or fulfilled \emph{in the past}~\citep{bartocci2018lectures}.
There are works on predictive monitoring, which assume that some abstract model of the system is either available~\citep{zhang2012runtime,pinisetty2017predictive} or can be learned at runtime~\citep{ferrando2023incrementally}. %, so that future violation or satisfaction of the specification can be predicted.
Our work is close to the latter, but, instead of learning a detailed general purpose model of the system, our TPM ``learns'' only the essential trend required to predict the future states in the vicinity of the current time. As discussed in Sec.~\ref{sec:intro}, time-to-collision (TTC) metric~\citep{vogel2003comparison} can be viewed as a special case of our method where only degree~1 Taylor's polynomial is used.

Numerical and data-driven inference algorithms are fundamental to many different disciplines, such as time series forecasting in economics~\citep{hyndman2018forecasting}, state estimation of dynamical systems from observed output sequences~\citep{diop1994interpolation,bunton2024confidently}, and data-driven online control~\citep{de2019formulas}.
While our data-driven prediction algorithm has some resemblance to existing techniques, to the best of our knowledge, our work is the first to apply such techniques to the setting of runtime monitoring of black-box controlled dynamical systems.

%The closest technique to our work is the time-to-collision (TTC) metric~\citep{vogel2003comparison}, which is routinely used~\citep{wang2021review} in autonomous driving to predict on-road safety violation. The details of TTC are discussed in Section~\ref{sec:intro}.
%TTC is defined as the time after which the vehicle will violate safety (i.e., cause ``collision'') if it continues with its current velocity. It is easy to see that TTC is a special case of our approach, where the Taylor's polynomial is restricted to only velocity---i.e., only the first order time-derivative of the positional state of the system.

%!TEX root=main.tex

\section{Preliminaries and Problem Statement}\label{prelims}

\paragraph{Controlled dynamical systems.}
A {\em controlled dynamical system}, or \emph{system} in short, is defined via
\begin{equation}\label{eq:dynamics}
 \frac{dx(t)}{dt} = f(x(t),u(t)), \quad x(0) = x_0, \quad \forall t\geq 0\;.\; u(t) = \pi(x(t)),
\end{equation}
where $t  \in \mathbb{R}_{\geq 0}$ denotes time, $x(t)\in \mathcal{X} \subseteq \mathbb{R}^n$ and $u(t)\in \mathcal{U} \subseteq \mathbb{R}^m$ denote the state and the control input at time $t$, $x_0\in \mathcal{X}$ is the initial state, $f\colon \mathcal{X} \times \mathcal{U} \rightarrow \mathcal{X}$ is the (nonlinear) dynamics, and $\pi: \mathcal{X} \rightarrow \mathcal{U}$ is the controller which assigns a control input to each state.
In what follows, we assume that the dynamics $f$ and the controller $\pi$ are locally Lipschitz continuous, and the Picard–Lindel\"of theorem guarantees the existence and uniqueness of the solution of \eqref{eq:dynamics}; the solution will be denoted as $\Sys\colon \mathbb{R}_{\geq 0}\to \mathbb{R}^n$ and called the \emph{trajectory} of the system.
Local Lipschitz continuity of dynamics is a standard assumption in control theory~\citep{DawsonGF23}, and is also satisfied by neural network controllers with all common activations, including ReLU, sigmoid, and tanh~\citep{SzegedyZSBEGF13}.

%Then, it is a classical result in dynamical system theory that the controlled dynamical system and the controller defined above gives rise to a unique {\em trajectory} $\Sys\colon \mathbb{R}_{\geq 0}\to \mathbb{R}^n$. For each time $t \in  \mathbb{R}_{\geq 0}$, $\Sys(t)$ denotes the state that the system at time $t$. The local Lipschitz continuity is a classic assumption in control theory~\citep{DawsonGF23}, and is also satisfied by neural network controllers with all standard activation function (e.g.~ReLU, sigmoid, tanh)~\citep{SzegedyZSBEGF13}.

\paragraph{Safety properties.} We consider {\em quantitative safety properties} which are functions of the form $\val\colon \mathcal{X}\to \mathbb{R}$, assigning a real valued \emph{safety level} to each system state. 
We say that the controlled dynamical system under a given controller {\em satisfies the safety property} if $\val(\Sys(t)) \geq 0$ for all time steps $t \in \mathbb{R}_{\geq 0}$, i.e.~if the safety level remains non-negative along the trajectory. For example, if we are interested in analyzing boolean (or {\em qualitative}) safety violations, we define $\val(x) = -1$ if the state $x$ is unsafe and $\val(x) = 0$ otherwise. If we are interested in a quantitative safety properties, like the value of a barrier function $B$~\citep{PrajnaJP07}, then we define $\val(x) = B(x)$ for all states~$x$.
	
\paragraph{Problem statement.} Suppose we are given a \emph{black-box} controlled dynamical system and a quantitative safety property $\val$; both the dynamics and the controller of the system are unknown but we can observe the resulting trajectory $\Sys$. 
Let $\tau\in \mathbb{R}_{\geq 0}$ be a given \emph{sampling time} and $h\in \mathbb{N}$ be a given prediction \emph{horizon}.
A \emph{predictive runtime monitor}, or a \emph{monitor} in short, observes the trajectory of the system at the sampling instances, and after each new observation, predicts the safety levels in the next $h$ sampling instances.
Formally, at each time $t\in \{\tau,2\tau,3\tau,\ldots\}$, the monitor takes the input sequence $\ldots,\Sys(t-2\tau),\Sys(t-\tau),\Sys(t)$ in account and outputs either the sequence $\val(\Sys(t+\tau)), \dots, \val(\Sys(t+h\tau))$ or a statistic thereof (e.g., the minimum $\val$ or the first instance when $\val$ becomes negative).
We consider the problem of designing a monitor for the given safety property.

\paragraph{Taylor's expansion.} Before presenting our monitor, we recall Taylor's polynomial of a $(l+1)$-times continuously differentiable function $g: \mathbb{R} \rightarrow \mathbb{R}$. For each $1 \leq i \leq l$, denote by $g^{(i)}$ the $i$-th derivative of $g$. For a fixed point $t \in \mathbb{R}$, the {\em Taylor's polynomial of $g$ of degree $l$} at point $t$ is 
\begin{equation}\label{equ:taylor's polynomial}
	P_l(s) = g(t) + \frac{g^{(1)}(t)}{1!}(s-t) + \frac{g^{(2)}(t)}{2!}(s-t)^2 + \ldots + \frac{g^{(l)}(t)}{l!}(s-t)^l.
\end{equation}
The following theorem is a classical result from mathematical analysis which provides an upper bound on the approximation error of a function via its Taylor's polynomial at a given point.
\begin{theorem}[Taylor's theorem~\citep{rudin1964principles}]\label{thm:taylor}
	Suppose that $g: \mathbb{R} \rightarrow \mathbb{R}$ is an $(l+1)$-times continuously differentiable function. Let $t \in \mathbb{R}$ and let $P_l$ be the Taylor's polynomial of $g$ of degree $l$ at point $t$. Then, for every $s \in \mathbb{R}$, there exists a point $r \in (t,s)$ such that
	\begin{equation*}
		g(s) - P_l(s) = \frac{g^{(l+1)}(r)}{(l+1)!}(s-t)^{l+1}.
	\end{equation*}
	Hence, if $B\geq \sup_{r\in (t,s)} |g^{(l+1)}(r)|$, then we have $|g(s) - P_l(s)| \leq \frac{B}{(l+1)!}(s-t)^{l+1}$.
\end{theorem}

\section{Algorithms}
\label{sec:algorithms}

The heart of our monitor is a numerical algorithm (Sec.~\ref{sec:prediction}) for predicting the future states of a given black-box system from the states observed so far along the trajectory.
These predicted future states will then be used to obtain the desired predictive runtime monitor (Sec.~\ref{sec:monitoring}).

\subsection{A Numerical Algorithm for Predicting Future States}\label{sec:prediction}

Our (state) prediction algorithm has two phases: (1)~In the {\em learning phase}, for each dimension $i\in [1;n]$ of the system's state space, we use a polynomial with time as its variable to approximate the dimension $i$ of the trajectory $\Sys$. The polynomial approximation function is obtained via a numerical procedure that uses a \emph{finite set} of past states $\Sys(t-k\tau),\ldots,\Sys(t-2\tau),\Sys(t-\tau),\Sys(t)$ with an appropriately large $k$, which we will refer to as the \emph{$\tau$-stencil of length $(k+1)$ ending at  $t$}.
We write $x_{-k} = \Sys(t-k\tau), \dots, x_{-1} = \Sys(t-\tau), x_0 = \Sys(t)$, omitting $t$ whenever it is clear from the context.
(2)~In the {\em prediction phase}, the obtained polynomial approximation functions are used to compute the predictions of the future states up to the horizon $h$, denoted as $x_1 = \Sys(t+\tau), \dots, x_h = \Sys(t+h \tau)$.

\paragraph{Learning phase.} The learning phase independently considers each dimension $i\in [1; n]$ of the system's state space and computes a polynomial approximation for the $i$-th dimension of the trajectory function $\Sys$. Hence, in what follows, without loss of generality we assume that $n = 1$ and present our procedure for computing a polynomial approximation to the real-valued signal $\Sys$.

We use the Taylor's polynomial $P_l$ of $\Sys$ of a given degree $l$ as the polynomial approximation. The challenge in obtaining $P_l$ is that its coefficients depend on the values of the derivatives $\Sys^{(1)}, \dots, \Sys^{(l)}$ at time point $t$, which are unknown to us owing to the black-box nature of the system.

Therefore, we numerically approximate the values of the derivatives using the backward difference~(BD) method~\citep{gear1967numerical} from the observed stencil $x_{-k}, \dots, x_0 $ ending at time $t$.
In particular, the BD approximation of the $i$-th derivative at $x_0$ is obtained as:
\begin{align*}
	\grad{i}{}{x_0} \coloneqq \begin{cases}
		(x_0-x_{-1})/\tau	&	\text{if } i=1,\\
		(\grad{i-1}{}{x_0}-\grad{i-1}{}{x_{-1}})/\tau	&	\text{otherwise}.
	\end{cases}
\end{align*}
The following closed-form expression can be obtained from the inductive definition above:
\begin{align}\label{equ:first order approximation of derivatives}
	\grad{i}{}{x_0} = \frac{\sum_{j=0}^i (-1)^j{i\choose j}x_{-j}}{\tau^i}.
\end{align}
It can be easily verified that $\grad{l}{}{x_0}$ depends on states up to $x_{-l}$ in the past, and therefore the length of the stencil must be $k+1\geq l$.
The approximation error is formally derived in the following lemma.
\begin{lemma}\label{lem:backward difference error}
	Let $i>0$ and suppose that $\Sys: \mathbb{R} \rightarrow \mathbb{R}$ is an $(i+1)$-times continuously differentiable function. Then, for every given stencil of length $k+1 \geq  i$, the following holds:  
	\begin{align}\label{eq:error of backward derivative}
		|\Sys^{(i)}(t)-\grad{i}{}{x_0}| \leq \tau\left|\frac{\Sys^{(i+1)}(t)}{(i+1)!}\sum_{j=0}^{i} (-1)^j{i\choose j}(-j)^{i+1}\right| + \mathcal{O}(\tau^{2}).
	\end{align}
\end{lemma}
The approximation $\grad{i}{}{x_0}$ is called the first order approximation, because for small $\tau<1$, asymptotically, the first order term in $\tau$ dominates the error (i.e., the error is $\mathcal{O}(\tau)$).
Higher order BD approximations would lead to smaller errors and will be considered in future works.
\begin{proof}[Proof of Lem.~\ref{lem:backward difference error}]
	In \eqref{equ:first order approximation of derivatives}, if we use the following infinite Taylor's series expansion of $x_{-j}$
	\begin{align*}
		x_{-j} = \Sys(t-j\tau) = \Sys(t) -\frac{j\tau}{1!}\Sys^{(1)}(t) + \frac{(j\tau)^2}{2!}\Sys^{(2)}(t) - \frac{(j\tau)^3}{3!}\Sys^{(3)}(t) + \ldots,
	\end{align*}
	we observe that terms with derivatives of $\Sys$ of order lower than $i$ cancel out, and we obtain:
	\begin{align*}
		\grad{i}{}{x_0} = \Sys^{(i)}(t) + \sum_{p=i+1}^\infty\frac{\tau^{p-i}\Sys^{(p)}(t)}{p!}\sum_{j=0}^{i} (-1)^j{i\choose j}(-j)^{p}. 
	\end{align*}
	The claim is established by separating the dominating term with $p=i+1$ (which results in the first order term in $\tau$) in the sum on the right hand side from the higher order terms.
\end{proof}

We use $\bar{P_l}(\cdot)$ to denote the \emph{approximated} Taylor's polynomial obtained by replacing each $\Sys^{(i)}(t)$ in \eqref{equ:taylor's polynomial} with its BD approximation $\grad{i}{}{x_0}$.

\begin{remark}
	The use of Taylor's polynomial is a design choice, and any other polynomial approximation could be used.
	Since the $l$-th degree polynomial that interpolates between $l+1$ points is unique, all approaches would provide the same answer. 
	It is possible to choose the stencil length larger than $l+1$, in which case the polynomial is no longer unique, and a ``best fit'' polynomial can be obtained, e.g., the one that minimizes the mean-squared error.
	We leave this for future work.
\end{remark}

\paragraph{Prediction phase.} The prediction phase of our monitor uses the approximated Taylor's polynomial $\bar{P_l}$ of the trajectory $\Sys$ around the current time $t$ in order to compute the predicted future states, denoted as $ \bar{x}_1 = \bar{P_l}(t+\tau), \ldots, \bar{x}_h = \bar{P_l}(t+h\tau)$.
The following theorem establishes an error bound.

\begin{theorem}\label{thm:prediction accuracy}
	Let $l>0$ and suppose that $\Sys: \mathbb{R} \rightarrow \mathbb{R}$ is an $(l+1)$-times continuously differentiable function. Let $\bar{P_l}(s)$ be the approximated Taylor's polynomial obtained from a given stencil of length $k+1\geq l$ ending at time $t$ and the given sampling time $\tau$.
	Let $h\in \mathbb{N}$ be a given horizon, and $m\in [1;h]$ be an arbitrary future sampling instance within the horizon.
	Suppose for every $p\in [1,l+1]$, $B_p$ denotes the upper bound on the $p$-th derivative of $\Sys$ in the interval $(t,t+m\tau)$, i.e., $B_p= \sup_{r\in (t,t+m\tau)}\Sys^{(p)}(r)$.
	Then,
	\begin{align*}%\label{equ:total prediction error in one step}
		|\Sys(t+m\tau) - \bar{P_l}(t+m\tau)| &\leq  \frac{B_{l+1}}{(l+1)!}(m\tau)^{l+1} + \tau\cdot\sum_{p=1}^l\left|\frac{B_{p+1}}{(p+1)!}\sum_{j=0}^{p} (-1)^j{p\choose j}(-j)^{p+1}\right| + \mathcal{O}(\tau^2)\\
		&= \mathcal{O}((m\tau)^{l+1} + \tau).
	\end{align*}
\end{theorem}

\begin{proof}
	Follows by combining Thm.~\ref{thm:taylor} and Lem.~\ref{lem:backward difference error}.
\end{proof}

Thm.~\ref{thm:prediction accuracy} suggests that, for small $m$ and for $\tau<1$, the prediction error is linear in $\tau$, i.e., $\mathcal{O}(\tau)$.
However, for long prediction horizons $h$ that allow $m \in [1;h]$ to become too large, the term $(m\tau)^{l+1}$ may dominate over $\tau$, and therefore the prediction error may increase to $\mathcal{O}((m\tau)^{l+1})$. Hence, considering larger prediction horizon $h$ requires using smaller sampling time $\tau$.
This trend is visible in the ablation tests that we present in the experiments section. Finally, we remark that the $(l+1)$-times differentiability assumption is necessary only for our error bound analysis in Thm.~\ref{thm:prediction accuracy}. However, our numerical algorithm for predicting future states remains well defined even without this assumption.

\subsection{Taylor-Based Predictive Monitoring (TPM)}\label{sec:monitoring}

	\setlength{\algomargin}{1.5em}
	\newcommand\mycommfont[1]{\footnotesize\ttfamily\textcolor{black!70!white}{#1}}
	\SetCommentSty{mycommfont}
\begin{algorithm2e}[t]
	\DontPrintSemicolon
	\LinesNumbered
	\SetAlgoLined
	\SetKwInOut{Parameter}{Input parameters}
	\KwIn{sampling time $\tau$, Taylor polynomial's degree $l$, horizon $h$, safety property $\val$}
		$Q\gets $ empty FIFO queue \tcp*[r]{will store the latest $l+1$ sampled states}
		\While{true}{
			$x \gets$ newly observed state \;
			$Q.\text{push}(x)$ \tcp*[r]{$x$ is added as the last element of $Q$}
%			\uIf{$Q.\text{size}()\geq l+1$}{
				\uIf{$Q.\text{size}()> l+1$}{ 
					$Q.\text{pop}$ \tcp*[r]{discard the oldest state to maintain $Q.\text{size}()=l+1$}
%				}
				$\est{x}{1},\ldots,\est{x}{h}\gets$ future states predicted from the stencil $Q[0],\ldots,Q[l]$ \label{line:state prediction} \tcp*[r]{Sec.~\ref{sec:prediction}} 
				\textbf{Output} $\{\val(\est{x}{1}),\ldots,\val(\est{x}{h})\}$  \emph{or} a desired statistic thereof \label{line:spec prediction} \tcp*[r]{e.g., the minimum}
			}
		}
		\caption{\textbf{T}aylor-Based \textbf{P}redictive \textbf{M}onitor (TPM)}
		\label{alg:monitor}
	\end{algorithm2e}

We now present TPM, our predictive runtime monitor for safety properties, based on the numerical state prediction algorithm in Sec.~\ref{sec:prediction}; the monitor's pseudocode is presented in Alg.~\ref{alg:monitor}. 
The monitor continuously observes the samples drawn from the system's trajectory and stores the latest $l+1$ samples in the FIFO queue $Q$.
Using the states stored in $Q$ as a stencil, TPM first computes the predicted future states using the approximate Taylor's polynomial (Line~\ref{line:state prediction}).
Afterwards, in Line~\ref{line:spec prediction}, it outputs the safety levels of the predicted states or a desired statistic thereof, like the minimum safety level in $h$ steps or the first time instance when safety is violated (i.e., $\val$ drops below zero).

\section{Experimental Evaluation}
\label{sec:experiments}
We implemented the algorithms from Sec.~\ref{sec:algorithms} in a prototype tool written in Python, and performed experiments to investigate the following two research questions: (i)~How does TPM compare to state-of-the-art TTC method in giving early warnings of safety violations?
(ii)~How accurate is TPM when compared against the ground truth data?
Both questions are studied on two environments.

%\begin{figure}
%	\centering\scalebox{0.7}{
%	\includegraphics[width=1.0\textwidth]{figures/2and3only.pdf}}\vspace{-1.5em}
%	\caption{Euclidean distance error for TTC and Static methods on the F1Tenth model.}\vspace{-1em}
%	\label{fig:f1tenth mse}
%\end{figure}
%\medskip
\noindent\emph{Environment 1: F1Tenth Racing~\citep{okelly2020f1tenth}.} A racing car of $1\!:\!10$ scale needs to drive around a track while avoiding getting too close to the track boundaries. The state vector comprises of the X-Y coordinate, the rotation, the forward velocity, and the angular velocity. The control inputs are the steering angle and the throttle. A state is safe if its distance to the track boundaries is greater than a predefined threshold, set to \SI{0.5}{\meter} meters in our experiments.
We consider $70$ differently parameterized controllers~\citep{kresse2024deep}, of which $54$ are the so-called Pure Pursuit~(PP) controllers which track a pre-planned path~\citep{coulter1992implementation} and 16 are Follow-The-Gap~(FTG) controllers which steer towards the direction where there is the most free space~\citep{sezer2012novel}.
The sampling time for this environment is fixed at $\tau = \SI{0.01}{\second}$.

%\medskip
\noindent\emph{Environment 2: F-16 Fighter Jet~\citep{DBLP:conf/adhs/HeidlaufCBB18}.} A simplified F-16 fighter jet system needs to fly at a safe height above the ground. The $16$-dimensional state vector comprises of air speed ($v_a$), angle of attack ($\alpha$), angle of sideslip ($\beta$), roll ($\phi$), pitch, yaw, roll rate, pitch rate, yaw rate, northward displacement, eastward displacement, altitude ($\mathit{alt}$), engine power lag, upward accel, stability roll rate, and slide accel and yaw rate. 
The $4$ control inputs are acceleration, stability roll rate, the sum of side acceleration and yaw rate, and throttle.
A state is safe if the altitude is between $1000\;\ft$ and $45000\;\ft$.
%We use randomly assigned initial states.
%For our experiments, we define a quantitative safety specification where the fighter jet remains within a target altitude between 1000ft and 45000ft. \KM{What is the quantitative aspect of this specification? To me it looks boolean (safety is ``true'' when the altitude is maintained, otherwise it is ``false.''}
The sampling time for this environment is fixed at $\tau = \SI{0.033}{\second}$.

%\medskip
%\noindent\emph{Environment 3: Ship.}
%The third environment is a ship control system as described in~\cite{DBLP:journals/ral/QinSF22}. The state vector is represented by $(x, y, \psi, u, v, r)^T$, where $(x, y)$ denotes the ship's position coordinates, $\psi$ is the heading angle, and $(u, v, r)$ correspond to the linear velocities in the longitudinal and lateral directions and the angular velocity, respectively. The control inputs are $(\tau, \Delta)$, where $\tau$ is the thrust and $\Delta$ the rudder angle. The primary objective of this control task is to avoid collisions with nearby obstacles while navigating the ship towards a randomly assigned goal state. Initial states are also randomly assigned. For our experiment, we used a neural-based controller obtained using the training script from ~\cite{DBLP:journals/ral/QinSF22}.

\subsection{Experiment 1: Comparison to the Classical Time-to-Collision Method}

We compared the performance of TPM to the baseline time-to-collision (TTC)~\citep{vogel2003comparison}, a widely used measure in autonomous driving for predicting on-road safety violations~\citep{wang2021review}. The TTC metric represents a \emph{special case} of our approach with $l=1$, utilizing only the current velocity to estimate the time to collision, assuming the vehicle continues in a straight line along its current orientation.

We design TPM monitors with two different types of outputs, namely boolean safety outputs and quantitative safety outputs, each of which can also be predicted by modifying the TTC algorithm.
%Since the focus is on safety, our monitor issues a warning whenever it detects an unsafe state within its predicted trajectory. 
%For the purposes of this evaluation, the actual trajectory observed by the monitor is treated as the ground truth.
%To this end, we define the following quality metrics: 
In the following, we describe the two types of outputs along with ways to measure their accuracy, for which we consider the actual system's trajectory as the ground truth.

\begin{wrapfigure}{r}{0.7\linewidth}
	\centering\scalebox{0.7}{
		\includegraphics[width=1.0\textwidth]{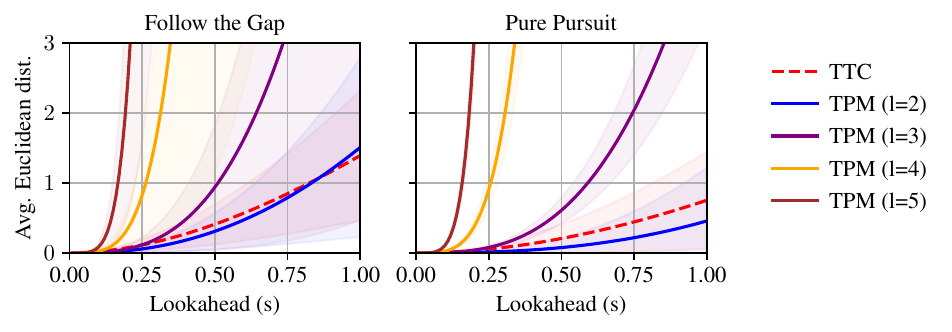}}
		\scalebox{0.7}{
		\includegraphics[width=1.0\textwidth]{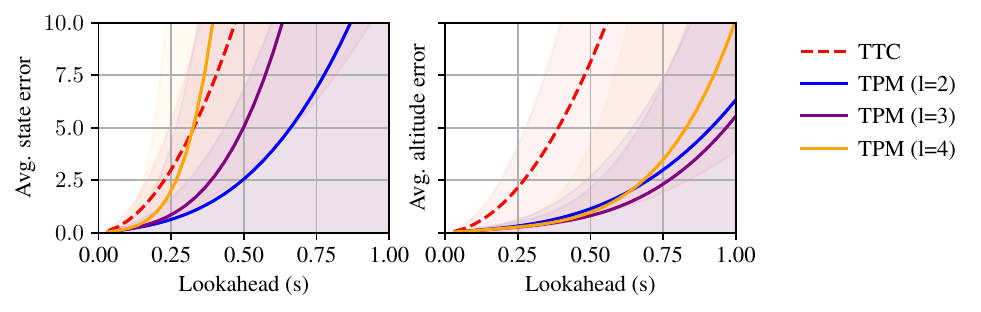}}\vspace{-1.2em}
	\caption{Ablation test results for prediction errors on F1Tenth (top row) and F-16 (bottom row). The lines represent the mean whereas the shaded regions represent the spread. For TPM, the constant $l$ is the degree of the Taylor's polynomial. The lookahead (X-axes) are measured as $h\tau$ for varying $h$. }\vspace{-1em}
	\label{fig:f1tenth mse2}
\end{wrapfigure}
\medskip
\noindent\textbf{Outputs with boolean accuracy metrics.} In this case, the monitor raises a warning if a safety violation is predicted within the prediction horizon. We measure the accuracy as follows: A warning is classified as a \emph{true positive} (TP) if it is issued prior to an unsafe state; otherwise, it is a \emph{false positive} (FP), indicating a false alarm. Conversely, if there was no prior warning but a unsafe state occurs, it is categorized as a \emph{false negative} (FN), representing a missed detection; otherwise, it is called a \emph{true negative} (TN). The true positive rate (aka, sensitivity) is defined as TPR = TP/(TP+FN), and the true negative rate (aka, specificity) is defined as TNR = TN/(TN+FP). 
TPR indicates how well the monitor can predict a real safety violation, whereas TNR indicates how well a monitor can predict the absence of it.

\medskip
\noindent\textbf{Outputs with quantitative accuracy metrics (Q).} In this case, the monitors' outputs are environment-specific. 
In the F1Tenth environment, the output of the monitor is the same as before, i.e., it raises a warning if an unsafe state is predicted within the prediction horizon. We measure accuracy as the earliest time before entering an unsafe state when a warning is issued: $Q_\textsc{F1Tenth}=(\min_{i\geq t_\textsc{unsafe}-1}t_\textsc{warning}^i)-t_\textsc{unsafe}$, where $t_\textsc{unsafe}$ is the time when the unsafe state takes place, $t_\textsc{warning}^i$ is the time $i$ when a warning is issued within the prediction horizon $h$.
For the F-16 environment, the monitor is required to output the \emph{minimum safety distance} within the horizon, defined as $d=\min_{i\in [0,h]} |(\est{x}{i}-\mathit{alt}_{min}) + (\mathit{alt}_{max} -\est{x}{i})|$. We measure accuracy as the difference between the minimum safety distance predicted by the monitor and the ground truth minimum safety distance observed: $Q_\textsc{F-16}=|d_\textsc{predicted}-d_\textsc{observed}|$, for a fixed prediction horizon set to $50$.

%\medskip
%\noindent\textbf{Baseline.} We consider the time-to-collision (TTC) metric~\citep{vogel2003comparison}, a widely used measure in autonomous driving for predicting on-road safety violations~\citep{wang2021review}. The TTC metric represents a \emph{special case} of our approach with $l=1$, utilizing only the current velocity to estimate the time to collision, assuming the vehicle continues in a straight line along its current orientation.
%	\begin{figure}\scalebox{0.7}{
%		\includegraphics[width=\textwidth]{figures/states_f110.pdf}}\vspace{-2em}
%	\caption{Prediction accuracy across varying prediction horizons for F1Tenth, with $l=2$. 
%		%The x-axis represents the time steps in the simulation, while the y-axis indicates the values for each state dimension. 
%		\label{fig:f1stateestimation}}\vspace{-1em}
%\end{figure}

\medskip
%We present the results in Tables~\ref{tab:performance_comparison_safety}, \ref{tab:performance_comparison1} and~\ref{tab:performance_comparison2}. 

The accuracy of the TPM and TTC monitors were measured on random simulations with random initial states of the two systems that we consider: for F1Tenth, we collected $775,300$ simulation steps with FTG controllers and  $2,647,570$ with PP controllers, and for F-16, we collected $33,750,000$ simulation steps.

The results are reported in Tab.~\ref{tab:performance_comparison1}.
We observe that, overall, TPM significantly outperforms TTC in all categories except F1Tenth with FTG controller, in which case, owing to the less smooth trajectories, TPM showed lower TNR, i.e., it was more ``cautious'' and more often predicted safety would be violated when in reality it did not. 
In all other cases, TPM was more accurate and showed higher TPR, TNR, and $Q$-value than TTC.

%We observe that TPM outperforms TTC, achieving significantly more true positives and fewer false positives in the F-16 environment. More specifically, our TPM method achieves an impressive \textbf{99.26\%} true positive rate and \textbf{0.74\%} false positive rate. The results are obtained from $75,000$ random rollouts of $450$ simulation steps starting with random initial states. In this example we did not observe FNs.
%For the quantitative metric, the TTC method yields a prediction error of 108.94 ft, whereas the TPM method reduces this error to 72.01 ft. We observe a similar trend in the F1Tenth environment, as summarized in Tables \ref{tab:performance_comparison1}. We ran a total of $775,300$ and $2,647,570$ simulation steps for FTG and PP agents respectively. In particular, the TPM method demonstrates exceptional performance for the PP agents for a prediction horizon of 50 steps. The smoother trajectories generated by PP agents allow the TPM method to make more accurate and reliable predictions, in contrast to FTG. As for the quantitative metric, the TPM method is able to signal warnings more effectively with a 
%longer lookahead of 100.

These findings are further confirmed by our experiments shown in Fig.~\ref{fig:f1tenth mse2}. Here, TPM significantly outperforms TTC in the F-16 environment as well as for the PP controllers in the F1Tenth environment. For FTG controllers, the TPM method performs better with a smaller lookahead, up to approximately 0.8 seconds (i.e., $h=80$).

%\begin{table}[ht]
%	\centering
%	\caption{F-16 environment: performance comparison between TTC and TPM}
%	\label{tab:performance_comparison_safety}
%	\small % Reduce font size
%	\begin{adjustbox}{max width=\textwidth}
%		\begin{tabular}{lcccc}
%			\toprule
%			\textbf{Metric} & \multicolumn{2}{c}{\textbf{TTC}} & \multicolumn{2}{c}{\textbf{TPM}} \\
%			\cmidrule(lr){2-3} \cmidrule(lr){4-5}
%			& \textbf{Value} & \textbf{(\%)} & \textbf{Value} & \textbf{(\%)} \\
%			\midrule
%			TP & 2518888 & 89.41\% & \textbf{2706320} & \textbf{99.26\%} \\
%			FP & 298477 & 10.59\% & \textbf{20309} & \textbf{0.74\%} \\
%			FN & 0 & 0.00\% & \textbf{0} & \textbf{0.00\%} \\
%			Q  & 108.94 ft & - & \textbf{72.01 ft} & - \\
%			\bottomrule
%		\end{tabular}
%	\end{adjustbox}
%	\vspace{-1em}
%\end{table}

\begin{table}%[h!]
	\centering
	\caption{Performance comparisons between TPM and TTC on the F1Tenth and F-16 environments. The ``$\%$'' in bracket in TP, FP, and FN are with respect to the total simulation step counts. The bold numerical entries indicate the which method among TPM and TTC was better.}
%	F1Tenth Environment: Performance Comparison of TTC and TPM. Results include for a fixed horizon $h$ of 50 and 100. %TODO: check this data
%		FN is recorded if no warning is issued within 25 and 50 steps respectively before entering an unsafe state.} 
\label{tab:performance_comparison1}\scalebox{0.95}{
	\begin{tabular}{lcccccc}
		\toprule
		Env. & $h$ & Metric & \multicolumn{2}{c}{\textbf{TPM}} & \multicolumn{2}{c}{TTC} \\
		\midrule
		\parbox[t]{2mm}{\multirow{13}{*}{\rotatebox[origin=c]{90}{F1Tenth}}} & & & FTG & PP & FTG & PP \\
		\cmidrule(lr){4-5} \cmidrule(lr){6-7}
		 & \multirow{4}{*}{50} & TP & \textbf{98323 (12.68\%)} & \textbf{246143 (9.30\%)} & 80023 (10.32\%) & 235180 (8.88\%) \\
		& & FP & 36329 (4.69\%) & \textbf{33049 (1.25\%)} & \textbf{12371 (1.60\%)} & 100915 (3.81\%) \\
		& & FN & \textbf{711 (0.09\%)} & \textbf{269 (0.01\%)} & 19011 (2.45\%) & 11232 (0.42\%) \\
		& & TPR & \textbf{0.993} & \textbf{0.999} &0.808&  0.954 \\
		& & TNR & 0.946 & \textbf{0.986} &\textbf{0.982} & 0.958 \\
		& & $Q$  & \textbf{47.28} & \textbf{49.69} & 39.22 & 46.61 \\
		\cmidrule(lr){2-7}
		 & \multirow{4}{*}{100} & TP & \textbf{98736 (12.74\%)} & \textbf{244257 (9.23\%)}  & 90232 (11.64\%) & 214305 (8.09\%) \\
		& & FP & 177354 (22.88\%) & \textbf{238744 (9.02\%)}  & \textbf{133671 (17.24\%)}  & 566401 (21.39\%) \\
		& & FN & \textbf{298 (0.04\%)}  & \textbf{2155 (0.08\%)}  & 8802 (1.14\%) & 32107 (1.21\%) \\
		& & TPR & \textbf{0.997} &\textbf{0.991} &0.911& 0.870\\
		& & TNR & 0.738 & \textbf{0.901} &\textbf{0.802} & 0.764\\
		& & $Q$ & \textbf{98.09} & \textbf{96.67}  & 90.93 & 84.38 \\
		\midrule
		\parbox[t]{2mm}{\multirow{6}{*}{\rotatebox[origin=c]{90}{F-16}}} & \multirow{6}{*}{50} & TP & \multicolumn{2}{c}{\textbf{2706320 (8.01\%)}}  & \multicolumn{2}{c}{2518888 (7.46\%)} \\
		& & FP & \multicolumn{2}{c}{\textbf{20309 (0.06\%)}}  & \multicolumn{2}{c}{298477 (0.88\%)} \\
		& & FN & \multicolumn{2}{c}{\textbf{0 (0.00\%)}}  & \multicolumn{2}{c}{\textbf{0 (0.00\%)}} \\
		& & TPR & \multicolumn{2}{c}{\textbf{1.00}} &\multicolumn{2}{c}{\textbf{1.00}}\\
		& & TNR &  \multicolumn{2}{c}{\textbf{0.989}} &\multicolumn{2}{c}{0.86}\\
		& & $Q$ & \multicolumn{2}{c}{\textbf{72.01 ft}}  & \multicolumn{2}{c}{108.94 ft} \\
		\bottomrule
	\end{tabular}}
\end{table}

%\begin{figure}
%	\centering\scalebox{0.8}{
%		\includegraphics[width=1.0\textwidth]{figures/mean_error_rmse_ft.pdf}}\vspace{-1.2em}
%	\caption{Ablation test for prediction error in F-16 environment.}\vspace{-1em}
%	\label{fig:f16_mse}
%\end{figure}

\subsection{Experiment 2: Prediction Accuracy and Ablation Tests}
To visually inspect the prediction accuracy of TPM, in Fig.~\ref{fig:f16stateestimation}, we plot the outputs of several instances of our monitor, with different horizon lengths, alongside the actual trajectory.
%In Fig.~\ref{fig:f16stateestimation} we visualize the trajectories along various dimensions for increasing horizons.
We observe that the prediction error increases with longer prediction horizons. %, whereas for shorter horizons, the predicted trajectory closely aligns with the true trajectory. 
To further analyze the relationship between prediction error, Taylor polynomial's degree $l$, and prediction horizon $h$, we conduct ablation tests whose results are shown in Fig.~\ref{fig:f1tenth mse2}.
We observe that for F1Tenth, the configuration with $l=2$ outperforms the rest, for both PP and FTG agents, with the improvement being particularly significant for PP agents.
As the prediction horizon increases beyond 0.8s (80 prediction steps), the configuration with $l=1$ becomes slightly better for the FTG agents. In the F-16 environment, we observe that $l=3$ yields the best performance.

Overall, prediction accuracy is correlated with the smoothness of the dynamical system and  the used controller. This is not surprising since from Thm.~\ref{thm:prediction accuracy}, we know that the prediction error increases as the values of the higher order derivatives increase, e.g., when the system has jerky movements.
\begin{wrapfigure}{r}{0.7\linewidth}
\centering
\scalebox{0.7}{
		\includegraphics[width=\textwidth]{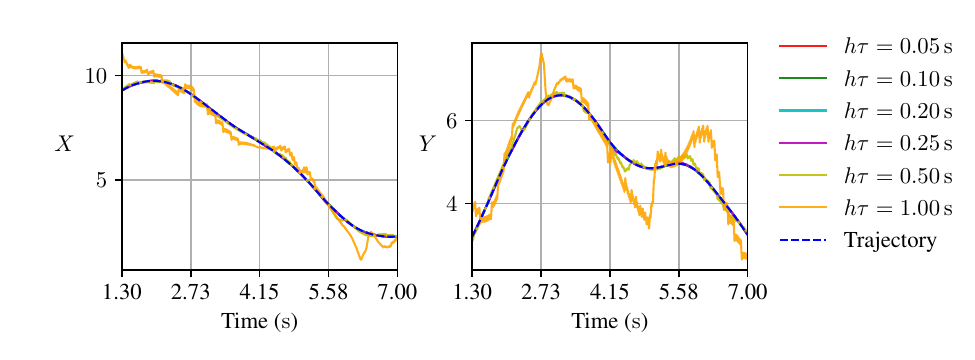}}\\ %\vspace{-1em}
\scalebox{0.7}{
		\includegraphics[width=\textwidth,trim={0 0 0 1.2cm},clip]{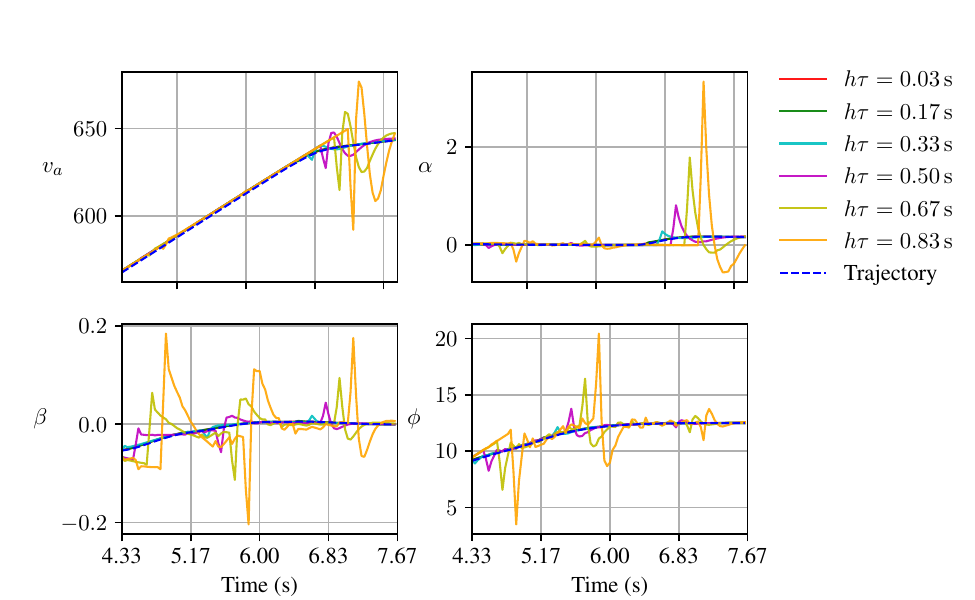}}
	\caption{Visualization of the monitor's output as compared to the ground truth trajectories for the first two state dimensions of F1Tenth with $l=2$ (top row) and four different state dimensions of F-16 with $l=3$ (middle and bottom rows). 
		%The x-axis represents the time steps in the simulation, while the y-axis indicates the values for each state dimension. 
		\label{fig:f16stateestimation}}\vspace{-1em}
	%The state is 16-dimensional. Here we visualize the first four dimensions with each plot corresponding to one dimension.  }
	\end{wrapfigure}	

%
%Figures~\ref{fig:f16_mse} and~\ref{fig:f1tenth mse2} illustrate the prediction error with varying history sizes and prediction horizons on the F-16 and F1Tenth models, respectively. To provide a more intuitive understanding of the prediction error, we study the average error in euclidean space for F1Tenth as well as in altitude for F-16. In Fig.~\ref{fig:f1tenth mse2}, the configuration with $l=2$ outperforms the rest, for both PP and FTG agents, with the improvement being particularly significant for PP agents.
%As the prediction horizon increases beyond 0.8s (80 prediction steps), the configuration with $l=1$ becomes slightly better for the FTG agents. In the F-16 environment, we observe that $l=3$ yields the best performance. Additionally, we provide the root mean square error (RMSE) over the entire state vectors to further evaluate prediction accuracy. Both metrics follow the same trend. Overall, prediction accuracy is correlated to both the smoothness of the dynamical system and the characteristics of the employed controller. They collectively determine the smoothness of the closed-loop system, which is also one of the most important underlying assumptions of our method.

\medskip
\noindent\emph{Monitoring overhead.} 
Our experiments were ran on a personal computer with 12th Gen Intel(R) Core(TM) i9-12900K processor and 32GB RAM.
In our experiments, the monitoring overhead per observation consistently remains below {0.001 seconds} for $l\leq 14$ and $h=1$, or for $l\leq 5$ and $h\leq 10$. For longer prediction horizons, the overhead incrementally increases to approximately 0.002 seconds. These results demonstrate the efficiency of our monitor, confirming its lightweight nature and suitability for practical applicability. %The monitor can thus be deployed online alongside the controller without disrupting control intervals.

\section{Conclusion}
We introduced a lightweight \emph{predictive} runtime monitoring framework for black-box controlled dynamical systems, which is able to predict safety violations ahead in time.
At each time step, our monitor learns a Taylor-based polynomial approximation of the system's state trajectory from the past observations, which is then used to perform predictions of future states so that safety violations can be predicted. We derive formal upper bounds on the prediction error, given the knowledge of bounds on the derivatives of the trajectory. We present the effectiveness of our monitor on models of a racing car and a fighter aircraft taken from the literature.
Future work will focus on studying numerical instabilities, higher order numerical approximations of the derivatives in Taylor's polynomial, different forms of polynomial approximations, as well as extensions to stochastic dynamical systems, multi-agent scenarios, and broader applications beyond safety verification.

%\todo{Future directions: LSE approximations, variable $\tau$, higher order approximations, formal analysis of numerical errors and their instability.}

\clearpage
\acks{This work was supported in part by the ERC project ERC-2020-AdG 101020093.}

\bibliography{refs}

\end{document}